\documentclass[11pt]{amsart}
\usepackage{amsmath,amsthm,amssymb}        
\usepackage[pdftex]{graphicx}
\usepackage[margin=1.3in]{geometry}

\newcommand{\M}  [1] {\ensuremath{{\overline{\mathcal M}}{_{0, #1}(\R)}}}   
\newcommand{\CM} [1] {\ensuremath{{\overline{\mathcal M}}{_{0, #1}}}}       

\def\<{\langle}
\def\>{\rangle}

\def\R{\mathbb{R}}

\newcommand{\vel}{\nu}
\newcommand{\cc}{cc}
\newcommand{\G}{\mathcal G}
\newcommand{\str}{\mathcal S}
\newcommand{\treemap} {{\Phi}}
\newcommand{\hide}[1]{}

\theoremstyle{plain}
\newtheorem{thm}{Theorem}
\newtheorem{prop}[thm]{Proposition}

\newtheorem{lem}[thm]{Lemma}

\newtheorem*{race}{Racing Lemma}
\newtheorem*{ang}{Angle Bisector Problem}
\newtheorem*{open}{Open Problem}

\theoremstyle{definition}
\newtheorem*{defn}{Definition}

\newtheorem*{note}{Notation}

\theoremstyle{remark}
\newtheorem*{rem}{Remark}
\newtheorem*{ack}{Acknowledgments}

\numberwithin{equation}{section}

\begin{document}

\title{Skeletal Rigidity of Phylogenetic Trees}

\author{Howard Cheng}
\address{H.\ Cheng: University of Arizona, Tucson, AZ 85721}
\email{howardc@email.arizona.edu}

\author{Satyan L.\ Devadoss}
\address{S.\ Devadoss: Williams College, Williamstown, MA 01267}
\email{satyan.devadoss@williams.edu}

\author{Brian Li}
\address{B.\ Li: Williams College, Williamstown, MA 01267}
\email{brian.t.li@williams.edu}

\author{Andrej Risteski}
\address{A.\ Risteski: Princeton University, Princeton, NJ 08544}
\email{risteski@princeton.edu}

\begin{abstract}
Motivated by geometric origami and the straight skeleton construction, we outline a map between spaces of phylogenetic trees and spaces of planar polygons.  The limitations of this map is studied through explicit examples, culminating in proving a structural rigidity result.
\end{abstract}



\maketitle

\baselineskip=17pt

%
%
\section{Motivation} 
\label{s:intro}

There has been tremendous interest recently in mathematical biology, namely in the field of phylogenetics.  The work by Boardman \cite{boa} in the 1970s on the language of trees from the homotopy viewpoint has kindled numerous structures of tree spaces.  The most notably could be that of Billera, Holmes, and Vogtmann \cite{bhv} on a space of metric trees.  Another construction involving planar trees is given in \cite{dm}, where a close relationship (partly using origami foldings) is given to \M{n}, the real points of moduli spaces of stable genus zero algebraic curves marked with families of distinct smooth points.  One can understand them as spaces of rooted metric trees with labeled leaves, which resolve the singularities studied in \cite{bhv} from the phylogenetic point of view.

As there exists spaces of planar metric trees, there are space of planar polygons:  Given a collection of positive real numbers $r = (r_1, \ldots, r_n)$, consider the moduli space of polygons in the plane with consecutive side lengths as given by $r$; this space can be viewed as equivalence classes of planar linkages.  There exists a complex-analytic structure on this space defined by Deligne-Mostow weighted quotients \cite{km}.  By considering the \emph{stable polygons} of this space, it quite unexpectedly becomes isomorphic to a certain geometric invariant theoretic quotient of the $n$-point configuration space on the  projective line \cite{hu}.  

Our goal, from an elementary level, is to construct and analyze a natural map between spaces of polygons and planar metric trees.  Given a simple planar polygon $P$, there exists a natural metric tree  $\str(P)$ associated to $P$ called its \emph{straight skeleton}.  It was introduced to computational geometry by Aichholzer et al.\ \cite{ah}, and used for automated designs of roofs and origami folding problems.  We consider the inverse problem:  Given a planar metric tree, construct a polygon whose straight skeleton is the tree.  

Section~\ref{s:pre} provides some preliminary definitions and observations, and a topological framework is provided in Section~\ref{s:top}.  The notion of velocity in capturing the skeleton of a polygon is introduced in Section~\ref{s:velocity}, and the main rigidity theorem is given in Section~\ref{s:rigid}:  For a phylogenetic tree $T$ with $n$ leaves, there exist at most $2n-5$ configurations of $T$ which appear as straight skeletons of convex polygons.  Section~\ref{s:race} contains the lemma which does the heavy lifting, which is analogous to the Cauchy arm lemma, used in the rigidity of convex polyhedra \cite[Chapter 6]{do}.  Finally, section~\ref{s:comp} closes with computational issues related to constructing the polygon given a tree, which also uncovers ties to a much older angle bisector problem. 

\begin{ack}
We thank Oswin Aichholzer, Erik Demaine, Robert Lang, Stefan Langerman, and Joe O'Rourke for helpful conversations and clarifications, and especially Lior Pachter for motivating this question.  We are also grateful to Williams College and to the NSF for partially supporting this work with grant DMS-0850577.
\end{ack}

%
%
\section{Preliminaries and Properties}  \label{s:pre}
\subsection{}

In this paper, whenever the term polygon is used, we mean a simple polygon $P$.   The \emph{medial axis} of $P$ is the set of points in its interior which are equidistant from two or more edges of $P$.  It is well known that if the polygon is convex, the medial axis is a tree \cite[Chapter 5]{do}: The leaves of this tree are the vertices of $P$, and the internal nodes are  points of $P$ equidistant to three or more sides of the polygon. 

If the polygon has a reflex vertex, however, the medial axis (in general) will have a parabolic arc.   The \emph{straight skeleton} $\str(P)$ of a polygon $P$ is a natural generalization of the medial axis, which constructs a straight-line metric tree for any simple polygon \cite{ah}:  For a polygon, start moving all of the sides of the polygon inward at equal velocity, parallel to themselves.  These lines, at each point of time, bound a similar polygon to the original one, but with smaller side lengths. Continue until the topology of the polygon traced out by this process changes. One of two events occur:
\begin{itemize} 
\item[1.] \emph{Shrink event}: \ When one of the original sides of the polygon shrinks to a point, two non-adjacent sides of the polygon become adjacent.  Continue moving all the sides inward, parallel to themselves again. 
\item[2.] \emph{Split event}: \ When one of the reflex vertices in the shrinking polygon touches a side of the polygon, the shrinking polygon is split into two.  Continue the inward line movement in each of them. 
\end{itemize}
The straight skeleton is defined as the set of segments traced out by the vertices of the shrinking polygons in the above process. Indeed, the straight skeleton is a tree, with the vertices of the polygon as leaves.  Figure~\ref{f:ss}(a) shows the example of the medial axis of a nonconvex polygon, along with its piecewise-linear straight skeleton in part (b).

\begin{figure}[h]
\includegraphics{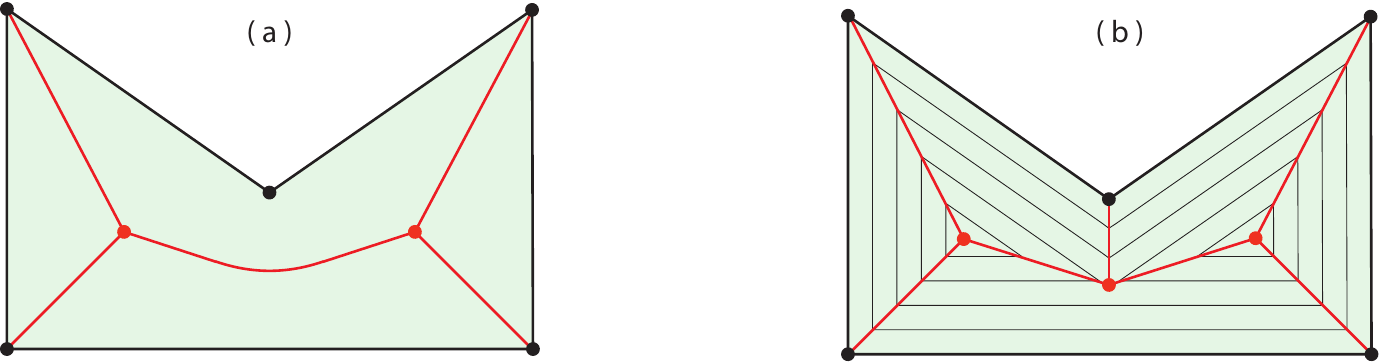}
\caption{(a) Medial axis and (b) straight skeleton.}
\label{f:ss}
\end{figure}

For convex polygons, the medial axis and the straight skeleton coincide. Furthermore, it will be useful for us to view the medial axis through this straight skeleton process lens. In this case, only shrink events occur. 

\subsection{}

Consider a map $\treemap$ from the space of simple polygons to the space of phylogenetic trees, defined as sending a polygon to its straight skeleton. In this section, we explore some basic properties of the map.  The reader is encouraged to consult \cite{euro}, where geometric details are given, and from which we inherit certain terminology.

\begin{defn}
Each edge of a \emph{phylogenetic tree} is assigned a nonnegative length, and each internal vertex has degree at least three.  A phylogenetic tree where the cyclic order of incident edges around every vertex is predefined is called a phylogenetic \emph{ribbon tree}.\footnote{This is sometimes called a \emph{fatgraph} as well \cite{pkwa}.}
\end{defn}

\begin{note}
Let $\G$ be the set of phylogenetic ribbon trees.  Let $E(G)$ denote a planar embedding of $G \in \G$ with straight-line edges, and respecting the predefined cyclic ordering around each vertex.  Moreover, define $P_{E(G)}$ to be the polygon resulting from connecting the leaves of $E(G)$ in cyclic order traversing around the tree.
\end{note}

\begin{defn}
A simple polygon $P_{E(G)}$ is \emph{suitable} if $E(G)$ equals its straight skeleton $\str(P_{E(G)})$, called a \emph{skeletal configuration} of $G$.  If there exists a suitable polygon for a tree $G \in \G$, we say $G$ is \emph{feasible}.
\end{defn}

We consider two natural examples of trees: stars and caterpillars.  A \emph{star} $S_n$ has $n+1$ vertices, with one vertex of degree $n$ connecting to $n$ leaves.  A \emph{caterpillar} becomes a path if all its leaves are removed.

\begin{prop}
There exist infeasible stars and caterpillars of $\G$. Thus $\treemap$ is not surjective.
\end{prop}

\begin{proof}
Consider a star $S_{3n}$ with edges $e_0, e_1, \ldots, e_{3n-1}$ in clockwise order. We set the edge length of $e_i$ to be equal to $x$ if $i \equiv 0$ mod 3, and equal to $y$ otherwise.  We claim that if $n \geq 3$ and the ratio $x/y$ is sufficiently small, then this tree cannot be the straight skeleton of any polygon; see Figure~\ref{f:inj}(a) for the case $n=3$.

\begin{figure}[h]
\includegraphics{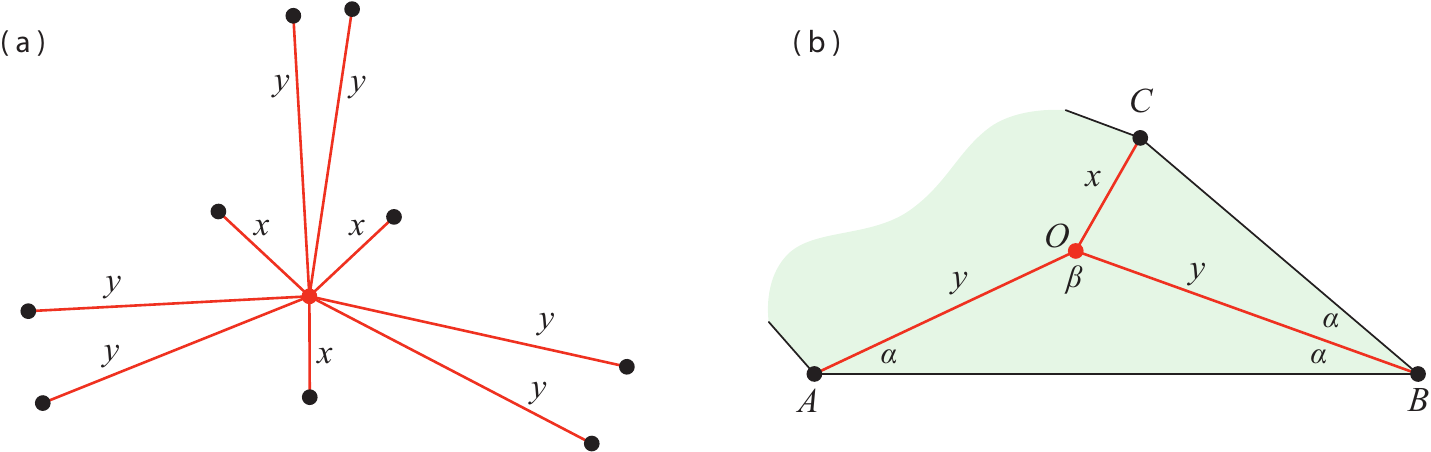}
\caption{Infeasible stars.}
\label{f:inj}
\end{figure}

Let $O$ be the center of the star, and $A, B, C$ denote leaves with edges $AO$ and $BO$ of length $y$ and $CO$ of length $x$;  see Figure~\ref{f:inj}(b).
In order to have a straight skeleton, edge $BO$ must bisect $\angle ABC$ into two angles of measure $\alpha$.  Defining $\beta := \angle AOB$, note that since triangle $AOB$ is isosceles, $\beta = \pi - 2\alpha$.  As the length of $x$ decreases relative to $y$, $\alpha$ becomes arbitrarily close to zero. Then $\beta$ must approach $\pi$; specifically, $\beta$ can be made greater than $2\pi/3$.  Since star $S_{3n}$ has $n \geq 3$ groups of three consecutive $\{x,y,y\}$ edges, then at least three such angles $\beta$ around the center $O$ are greater than $2\pi/3$, giving a contradiction.

\begin{figure}[h]
\includegraphics{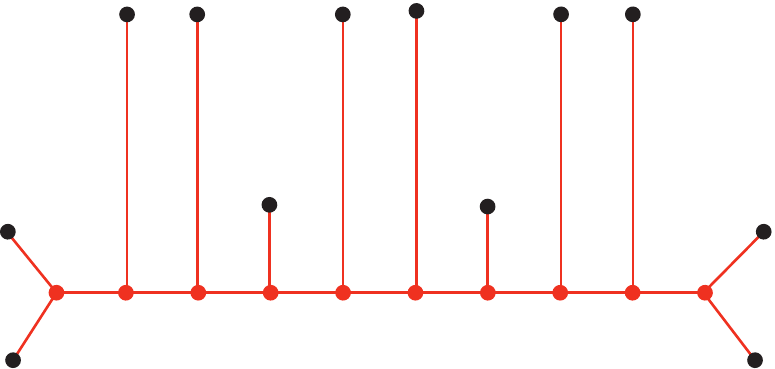}
\caption{Infeasible caterpillars.}
\label{f:inj-proof}
\end{figure}

The stars above can be tweaked to show infeasible caterpillars, as given in Figure~\ref{f:inj-proof}, with a nearly identical proof, taking internal edges to be arbitrarily small.
\end{proof}

\begin{prop} 
There are feasible caterpillars for which multiple suitable polygons exists.  Thus, $\treemap$ is not injective.
\end{prop} 

\begin{proof} 
Consider the tree $G$, as drawn in Figure~\ref{f:counter-sur-tree}, with two edge lengths $x$ and $y$, with $y$ substantially larger than $x$.  
\begin{figure}[h]
\includegraphics{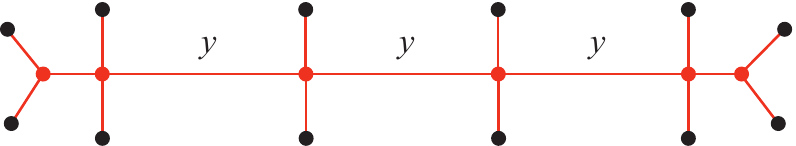}
\caption{Caterpillars with multiple suitable polygons.}
\label{f:counter-sur-tree}
\end{figure}
There exist multiple orthogonal polygons with straight skeletons equivalent to $G$, as in Figure~\ref{f:counter-sur-poly}.  We can divide a polygon into two pieces by cutting along a line perpendicular to a skeleton edge of length $y$.  Fixing one side of the divided polygon, reflecting the other side, and then gluing results in another polygon with the same straight skeleton.  Note that this transformation can be performed at each skeleton edge of length $y$.
\end{proof}

\begin{figure}[h]
\includegraphics{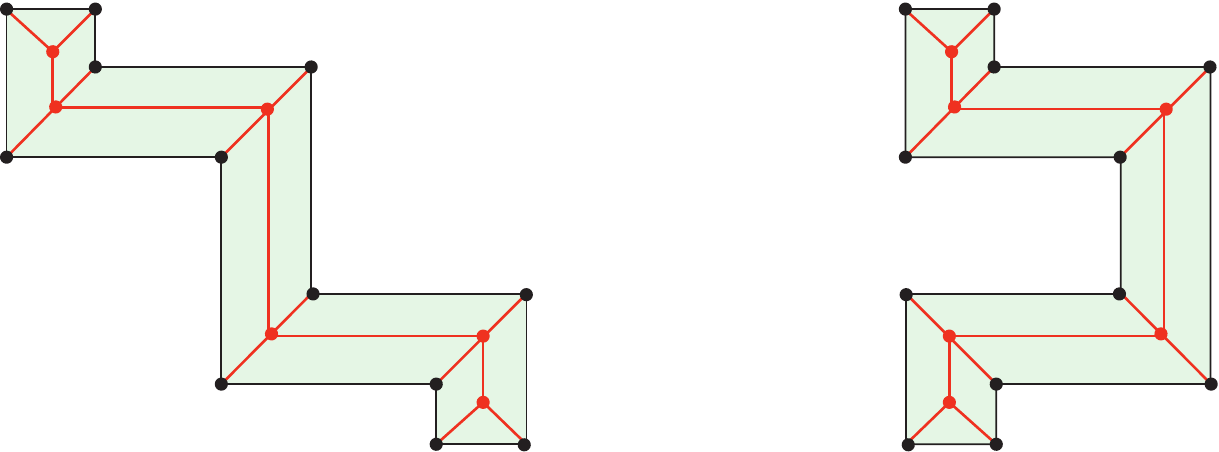}
\caption{Distinct suitable polygons for the same caterpillar graph.}
\label{f:counter-sur-poly}
\end{figure}

%
%
\section{Tree Topology} \label{s:top}

We now consider questions related to the deformations of phylogenetic trees and their respective polygons from a topological perspective.  Indeed, as one would expect, disregarding the edge lengths trivializes numerous questions.  We say two trees in $\G$ have the same \emph{topology} if they are equivalent as abstract graphs.  
If we consider $\treemap$ as a map between polygons and \emph{topological} trees, then the following shows the surjectivity of this map. 

\begin{prop}
\label{p:top-feasible}
For any $G \in \G$, there exists a feasible tree with the topology of $G$.
\end{prop}

\begin{proof}
Let $r$ be an internal vertex of $G$, and let $v_1, \ldots v_n$ denote the vertices adjacent to $r$.  Associate each vertex of a regular $n$-gon $P$ to a unique $v_i$, preserving the cyclic ordering in $G$.  We algorithmically modify $P$ so that its straight skeleton has the topology of $G$, as follows. Suppose that some $v_i$ is adjacent to $k$ other vertices $s_1, \ldots, s_k$ of $G$ (of depth 2). Then replace the vertex of $P$ associated to $v_i$ with $k$ new vertices forming $k-1$ sufficiently small edges in such a way that the angle bisectors of the $k$ vertices all intersect at a point. One can view this as `truncating' the vertex to create $k$ new vertices (see Figure~\ref{f:truncation}); this is possible since the truncations can be arbitrarily small.  Associate each new vertex of the truncated polygon to each $s_i$, again preserving the cyclic ordering of the vertices. This truncation procedure is repeated until the process terminates, that is, until each vertex of the truncated polygon is associated to a leaf in $G$. The final polygon then has a straight skeleton with the same topology as $G$. 
\end{proof}

\begin{figure}[h]
\includegraphics[width=\textwidth]{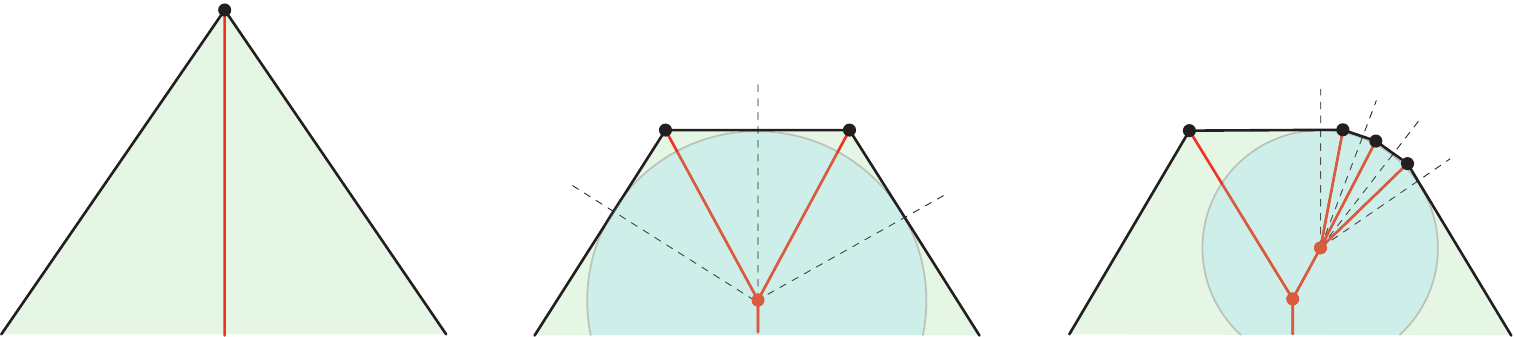}
\caption{Iterative truncation of a polygon resulting in a feasible tree.}
\label{f:truncation}
\end{figure}

For a polygon $P$ with $n$ edges, we may embed $P$ in the plane and let $(x_i,y_i)$, $1 \leq i \leq n$ denote its vertices in $\R^2$. So we can associate $P$ to a point $\gamma(P)$ in $\R^{2n}$, where
$$\gamma(P) \ = \ (x_1,y_1,x_2,y_2,\ldots,x_n,y_n).$$
Two polygons $P_1$ and $P_2$ are \emph{isotopic} if there exists a (continuous) isotopy $f:[0,1] \longrightarrow \R^{2n}$ such that $f(0)=\gamma(P_1)$, $f(1)=\gamma(P_2)$, and for every $t \in [0,1]$, $f(t)=\gamma(P(t))$ for some simple polygon $P(t)$, and some appropriate ordering of the vertices.
The following shows that \emph{convex} polygons with identical event chronologies in the preimage of a topological tree under $\treemap$ are in the same isotopy class.
 
\begin{thm}
\label{t:top-skeleton}
Let $P_1$ and $P_2$ be \emph{convex} polygons having identical event chronologies, with topologically equivalent straight skeletons.  Then there exists an isotopy between $P_1$ and $P_2$ such that $\str(P(t))$ is topologically equivalent to $\str(P_1)$ and $\str(P_2)$, for all $t$.
\end{thm}

\begin{proof}
We proceed by induction on the number of vertices of $P_1$ and $P_2$. The base case on three vertices is easily checked, so assume the statement holds for polygons on $n$ vertices, and suppose that $P_1$ and $P_2$ have $n+1$ vertices. Observe that the shrinking process that defines the straight skeleton is an isotopy of polygons (under a suitable time reparameterization), until an event occurs. Since $P_1$ and $P_2$ are convex, the only events which can occur are shrink events, with a polygonal edge shrinking to zero.  When these events occur, the number of vertices of the shrinking polygon decreases by at least one, so we obtain convex polygons $P_1'$ and $P_2'$ on $n$ vertices, whose straight skeletons have identical event chronologies and the same topologies. 

By the induction hypothesis, there is an isotopy $f'$ between $P_1'$ and $P_2'$ fulfilling the desired conditions. Now consider the isotopy $g_1$ obtained from stopping the shrinking processes for $P_1$ at some time $\varepsilon_1$ before the first event occurs; see Figure~\ref{f:cont-top}(a) to (b). Similarly, we let $g_2$ denote the isotopy from stopping $P_2$ at some time $\varepsilon_2$ before the first event; see Figure~\ref{f:cont-top}(d) to (c). For sufficiently small $\varepsilon_1$ and $\varepsilon_2$, the vertices of the polygons at the ends of the isotopies $g_1$ and $g_2$ become arbitrarily close to the vertices of $P_1'$ and $P_2'$. 
We can then use $f'$ to construct an isotopy $f''$ (from part (b) to (c) of the figure) such that $g_2^{-1} \circ f'' \circ g_1$ is the desired isotopy from $P_1$ to $P_2$ (up to time reparameterization) satisfying the conditions of the theorem.
\end{proof}

\begin{figure}[h]
\includegraphics[width=\textwidth]{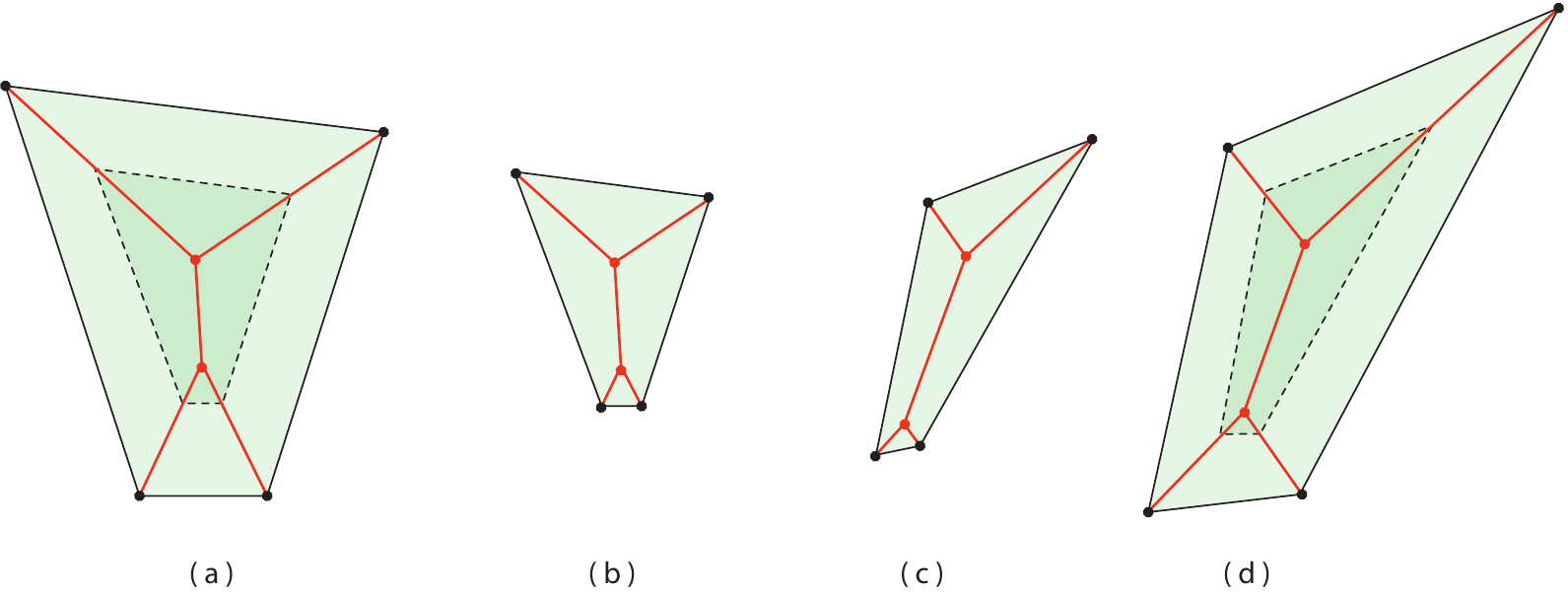}
\caption{Isotopic deformation of the straight skeletons.}
\label{f:cont-top}
\end{figure}

%
%
\section{Velocity Framework}  \label{s:velocity}

The following sections consider an appropriate concept of rigidity for phylogenetic trees.  Recall that a skeletal configuration of a phylogenetic tree $G \in \G$ is a planar embedding of $G$ such that it is the straight skeleton of the polygon that is determined by the leaves.

\begin{defn} 
A skeletal configuration $E(G)$ of a tree $G \in \G$ is $\emph{rigid}$ if, in the space of deformations of planar embeddings of $G$, $E(G)$ is an isolated skeletal configuration.
\end{defn}

\noindent Our main result, Theorem~\ref{t:rigidity_convex}, shows a condition stronger than rigidity when restricted to the case of convex polygons.  Our method of attack is to consider an alternative method based on \emph{velocity} to interpret how the straight skeleton is constructed:  Since each edge $e$ of the straight skeleton is traced out by the polygonal vertices, a velocity can be assigned to the vertex tracing out edge $e$. 

We start by applying this idea to the simplest types of trees:   Consider a star $S_n \in \G$, consisting of a center vertex $O$, and edges $e_i$ incident to leaves $v_i$.  During the straight skeleton construction, assuming each polygonal edge moves inward with unit velocity, each vertex $v_i$ of the polygon can be assigned a velocity $\vel_i$. Then the amount of time vertex $v_i$ requires to traverse its edge $e_i$ is $l(e_i)/\vel_i,$ where $l(e_i)$ is the length of edge $e_i$. Since $S_n$ is a star, all vertices start and end their movement during the skeletal construction at the same time, meeting at the center $O$. Thus 
\begin{equation}
\label{e:length-vel}
\frac{l(e_i)}{\vel_i} \ = \ \frac{l(e_1)}{\vel_1},
\end{equation}
for each $2 \leq i \leq n$, establishing the relative velocities of all the vertices.

There is also a useful relationship between the angles subtended at the vertices and their speed.  Assume that the angle at vertex $v$ is $2 \alpha$ if the angle at $v$ is convex, or $2 \pi - 2 \alpha$ if it is reflex.  If the sides incident on $v$ have been moving for time $t$ at unit speed, they will have traversed $t$ units, with vertex $v$ reaching another point $w$.  From Figure~\ref{f:angle}, it follows that $t(\sin \alpha)^{-1}$ is the length of $vw$, implying the velocity  of vertex $v$ to be  $\vel \ = \ (\sin \alpha)^{-1}.$
\begin{figure}[h]
\includegraphics{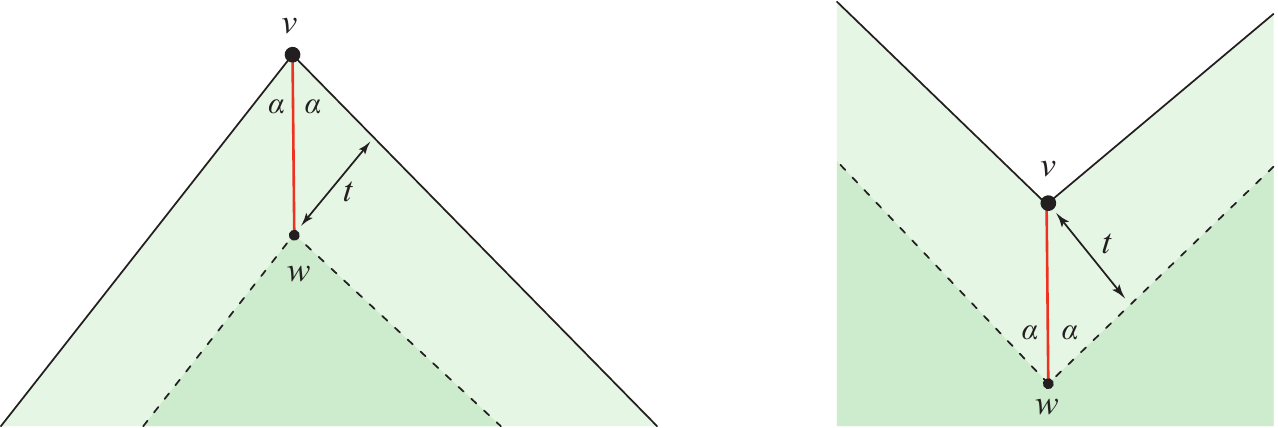}
\caption{Relationship between angle and speed.}
\label{f:angle}
\end{figure}

\begin{rem}
If $\alpha \leq \pi/2$ (convex), then $\alpha = \arcsin (1/\vel)$.  Otherwise $\alpha > \pi/2$ (reflex) and $\alpha = \pi - \arcsin(1/\vel)$. 
Since, a priori, we do not know whether the vertex angles are convex or reflex, we denote $\arcsin^*\theta = \{\arcsin \theta, \pi - \arcsin \theta\}$, to alleviate notational hassle.
\end{rem}

We now use the above observations to obtain the following:

\begin{prop} 
\label{p:star-rigid}
A star graph $G$ has a finite number of skeletal configurations.  In other words, there are a finite number of suitable polygons for $G$. 
\end{prop} 

\begin{proof}
The sum of the internal angles of a polygon with $n$ vertices is $(n-2) \pi$. So if the angle subtended at vertex $v_i$ for a suitable polygon of $G$ is $2 \alpha_i$, then 
$$\alpha_1 \ + \ \cdots \  + \ \alpha_n \ = \ \frac{(n-2) \cdot \pi}{2}.$$
\hide{
Equivalently, this can be written as one of few equations 
\[
\arcsin^* \frac{1}{v_1} + \arcsin^* \frac{1}{v_2} + \dots + \arcsin^* \frac{1}{v_n} = \frac{(n-2) \pi}{2} \,
\]
where $\arcsin^* s$ can be either $\arcsin s$ or $\pi - \arcsin s$. 

Indeed, the angles $\alpha_{i}$ in our polygon are in the range $[0,\pi]$. Furthermore, if $\alpha_{i} \leq \frac{\pi}{2}$, then $\frac{1}{v_i} = \sin \alpha_i$ and $\alpha_i = \arcsin \frac{1}{v_i}$. If, on the other hand, $\alpha_{i} > \frac{\pi}{2}$, $\frac{1}{v_i} = \sin (\pi - \alpha_i)$ and $\alpha_i = \pi -\arcsin \frac{1}{v_i}$. Notice that a priori we don't know whether the angles at the vertices are reflex or vertex - and of course, some of these systems might not give rise to solutions. We will, however, prove, each possible choice for $\arcsin^*$ leads to at most a finite number of solutions - which is all we need. 
}
In light of Eq.~\eqref{e:length-vel}, this can be rewritten as 
\begin{equation}
\label{e:arcsin}
\arcsin^* \left(\frac{1}{v_1}\right) + \arcsin^* \left(\frac{l(e_1)}{l(e_2) \cdot v_1}\right) + \cdots  + \arcsin^* \left(\frac{l(e_1)}{l(e_n) \cdot v_1}\right) \ = \ \frac{(n-2)\pi}{2}.
\end{equation}
Since the lengths of the star edges are fixed, this can be viewed as an equation in $1/\vel_1$.  We show this equation, for each possible choice for $\arcsin^*$, leads to at most a finite number of solutions.  

First assume that once Eq.~\eqref{e:arcsin} is rewritten as
$$\phi \ := \ \arcsin (m_1 x) \ \pm \ \arcsin (m_2 x) \ \pm \ \dots \ \pm \ \arcsin (m_n x) \ = \ c,$$
the constant term $c$ is nonzero. Since the Maclaurin series expansion for $\arcsin$ is 
\[
\arcsin z \ = \ \sum_{k=0}^\infty \ \frac{(2k)!}{2^{2k} \ (k!)^2 \ (2k+1)}   \ z^{2k+1} \, ,
\]
then $\phi$ also has an infinite series expansion when all of the $m_i x$ terms are between -1 and 1. Hence, if $\phi - c$ has an infinite number of solutions on a bounded interval for $x$, it must be identically zero. But notice that $\arcsin (m_i x)$ has no constant term in the infinite series expansion. Because $c$ is nonzero, $\phi - c$ cannot be identically zero, a contradiction.  Thus there are at most a finite number of solutions in this case.

Now assume the constant term $c$ vanishes.  One can show that at least one $\arcsin (m_i x)$ term will remain in Eq.~\eqref{e:arcsin}.  Out of all such terms, consider the one with the maximum $m_i$ value, say $\widehat m$.  Looking at the Maclaurin series expansion again, for sufficiently large $k$, the coefficient in front of $z^{2k+1}$ will be dominated by 
\[
\frac{(2k)! \ \widehat m^{2k+1}}{2^{2k} \ (k!)^2 \ (2k+1)}
\]
since $\widehat m$ must be positive and strictly larger than all other $m_i$ values.  Thus, for sufficiently large $k$, the coefficient in front of $z^{2k+1}$ will not vanish, and therefore $\phi - c$ cannot be identically zero, a contradiction. 
\end{proof}

\begin{rem}
In \cite[Lemma 8]{euro}, every caterpillar graph is shown to have a finite number of skeletal configurations, using different proof techniques.
\end{rem}

%
%
\section{Racing Lemma} \label{s:race}

Let $P$ be a convex polygon, and $P(t)$ denote the polygon formed by the edges of $P$ moving inwards at unit speed at time $t$. Suppose that there are a total of $n$ events which occur at times $t_1 < \ldots < t_{n}$. Observe that $P(t_1), \ldots,  P(t_{n-1})$ forms a sequence of polygons with a strictly decreasing number of vertices, resulting in the point $P(t_n)$.  We call $P(t_n)$ the \emph{chronological center} of polygon $P$. Figure~\ref{f:cc}(a) shows an example, where the direction on the edges of $\str(P)$ is based on how the skeleton is constructed.  The unique sink of this directed graph is the chronological center.  We omit the proof of the following:

\begin{lem} 
\label{l:chrcenter}
Let $P$ be a convex polygon.  If $P$ is in general position, then $\cc(P)$ is a vertex of $\str(P)$; otherwise, $\cc(P)$ is either a vertex or an edge of $\str(P)$.
\end{lem}

\begin{figure}[h]
\includegraphics[width=\textwidth]{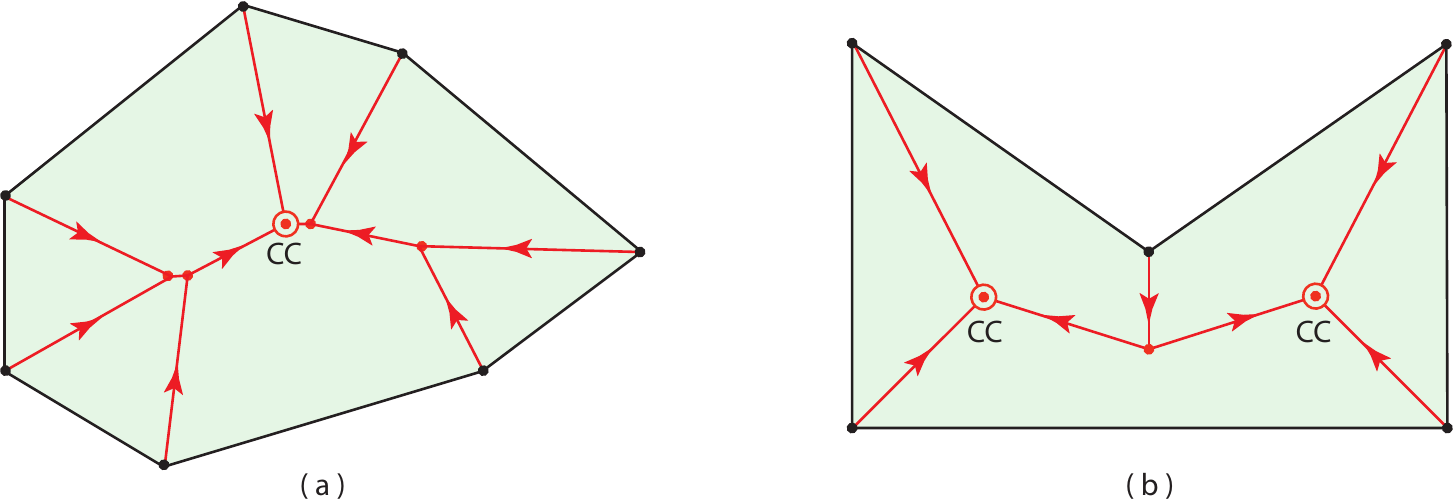}
\caption{Chronological centers of polygons.}
\label{f:cc}
\end{figure}

\begin{rem}
It is worth noting that the chronological center can be generalized for nonconvex polygons as well:  Here, multiple sinks will appear, one for each polygon that shrank to a point during an event, such as in Figure~\ref{f:cc}(b).
\end{rem}

Our key lemma is an analog of Cauchy's Arm Lemma in the theory of polyhedra reconstruction: The arm lemma states that if we increase one of the angles of a convex polygonal chain, the distance between the endpoints will only increase. To parallel this, 
we show that increasing one of the velocities of the leaves causes all velocities in the tree to be increased. More specifically: 

\begin{race} 
Let $G \in \G$ be the skeletal configuration of a suitable \emph{convex} polygon $P$.  If we increase the velocity of one of the leaves a sufficiently small amount, the velocities of all nodes of $G$ (other than chronological center) must increase in order for $G$ to remain a skeletal configuration.
\end{race} 

\begin{proof}
Root $G$ at its chronological center $\cc(G)$, which we assume to be a vertex of $G$.\footnote{The case when $\cc(G)$ is not a vertex is considered later.}
Notice that this uniquely determines the order in which the shrink events occur in the tree. We  prove a slightly stronger claim, namely that for any subtree with a root distinct from the chronological center, increasing the velocity of one of the leaves in the subtree an arbitrarily small amount forces the velocities of all nodes in the subtree to increase. 

Proceed by induction on the maximum depth of the subtree, where the depth is the defined as the maximum topological length of a path from the root to a leaf in the subtree.  If the maximum depth of the subtree is one, we have a group of leaves with a common parent. Let the leaves have velocities $\vel_1, \ldots, \vel_k$. If the lengths of the edges from the parent to the leafs are $l_1, \ldots, l_k$, correspondingly, then $l_i/\vel_i = l_j/\vel_j.$   Thus, if the velocity of one leaf increases, so must all others.  

Now consider a subtree of depth $k$ with root $O$, and let the children of $O$ be $O_1, \ldots O_m$.  
Since $G$ is a skeletal configuration of $P$, the subtree of $O$ corresponds to $P$ being ``chiseled out'' by two supporting lines $AB$ and $AC$ as shown in Figure~\ref{f:ind-2}. 
\begin{figure}[h]
\includegraphics[width=.7\textwidth]{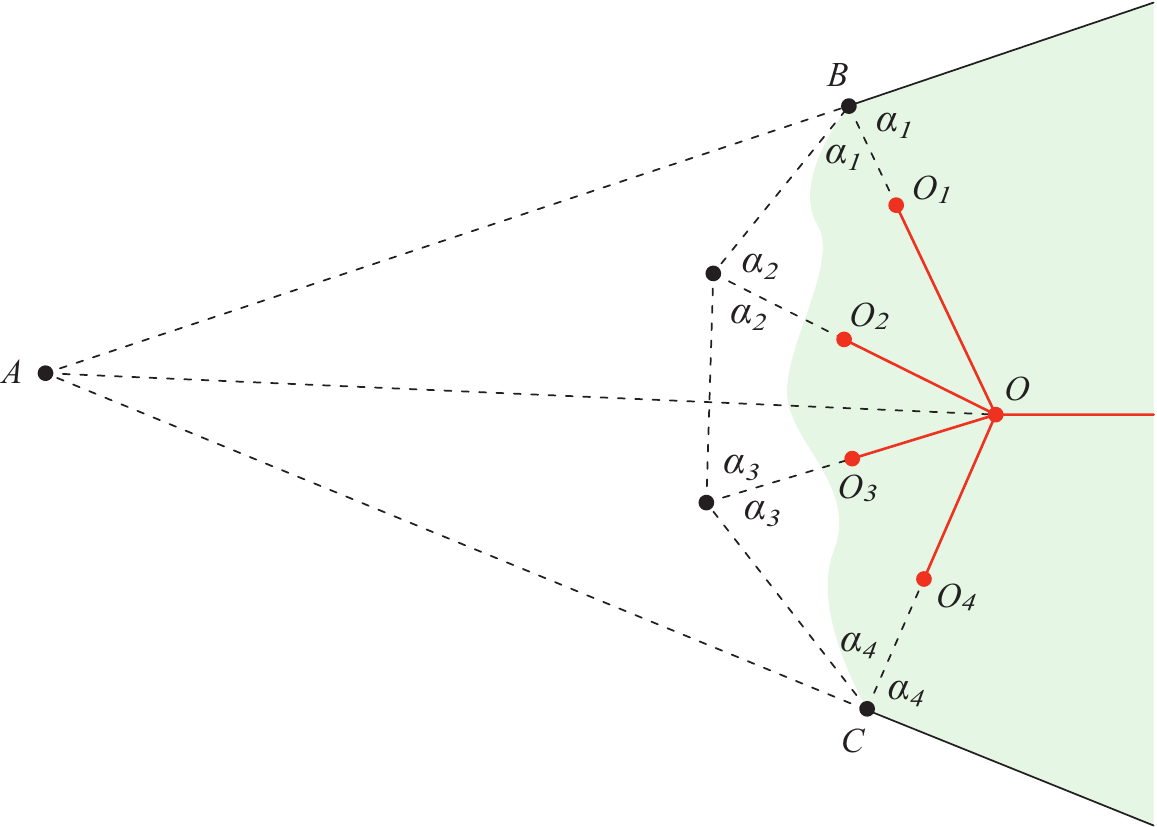}
\caption{Tree to polygon perspective for the Racing Lemma.}
\label{f:ind-2}
\end{figure}
Here, the sequence of edges lying between the edges of the polygon corresponding to lines $AB$ and $AC$ will have shrunk before the occurrence at node $O$, 
where the lines of $OO_i$ are angle bisectors of some (possibly non-adjacent) edges of $P$. 
Since $AO$ is an angle bisector of the angle between lines $AB$ and $AC$, arithmetic shows that 
$$\psi \ := \ \frac{\pi (m-1) - (\pi - 2 \alpha_1) - (2\pi -2\alpha_2) - \cdots - (2\pi - 2\alpha_{m-1}) - (\pi - 2 \alpha_m)}{2}$$ 
equals $\angle BAO$ and $\angle CAO$.  Now if $O_i$ has velocity $(\sin \alpha_i)^{-1}$, then $O$ has velocity $\arcsin \psi$.  Since the polygon is convex, however,  $\alpha_i$ and $\psi$ are convex angles, where $\sin$ is monotonically increasing. Therefore, if the velocities $(\sin \alpha_i)^{-1}$ of nodes $O_i$ increase, so does the velocity of $O$. 

Assume we increase the velocity of a leaf in the subtree of $O_1$. Then, by the inductive hypothesis, all of the vertices in $O_1$'s subtree (including $O_1$) will increase in velocity, so that $O_1$ finishes tracing out edge $O_1O$ faster than before.  For any other child $O_i$ of $O$, since the edge $O_iO$ is traced out at the same time as $O_1O$, the velocity of $O$ also increases.
\end{proof}

\begin{rem}
This lemma can be strengthened to include convex polygons in \emph{degenerate} positions, where from Lemma~\ref{l:chrcenter}, the chronological center can be an edge, with endpoints $T_1$ and $T_2$.  If we increase the velocity of a leaf in tree $T_1$, all of the vertices of the nodes of $T_1$ will increase.  But since both endpoints will be reached at the same time, it follows that all nodes in $T_2$'s subtree must increase their velocities.
\end{rem}

%
%
\section{Convex Rigidity} \label{s:rigid}

One more lemma is needed in order to prove the main rigidity result:

\begin{lem} 
Let $G \in \G$, with a fixed chronological center, and an assignment of velocities to each leaf.  Then there is at most one suitable \emph{convex} polygon.
\label{l:maxone}
\end{lem}

\begin{proof}
This is based on induction on the number of edges of tree $G$. The claim is true for any tree with three edges by Proposition~\ref{p:triangleunique}, so the base case is covered. Now consider a phylogenetic tree with $k$ edges.   We will use the relationship between the angles subtended at the vertices and their speed. 
Assume there are two distinct, convex polygons $P$ and $Q$ with the same angles (corresponding to velocity assignments), both with $G$ as their skeleton, with identical chronological centers.  Label the vertices of $G$ the same for both $P$ and $Q$, but bear in mind $G$ has different embedding for the two polygons.

Let the first event be the simultaneous shrinking of edges $A_1 A_2$, \ldots, $A_{m-1} A_m$. Since the velocities of the leafs are equal, and the chronological center is the same, this same event must happen first for both polygons. (Remember, rooting the tree in the chronological center uniquely determines the event sequence.) Let $O$ be the parent of leaves $A_1, \ldots, A_m$ in $G$.  Notice that all of the triangles $OA_iA_{i+1}$ must be congruent in $P$ and $Q$.

\begin{figure}[h]
\includegraphics[width=.7\textwidth]{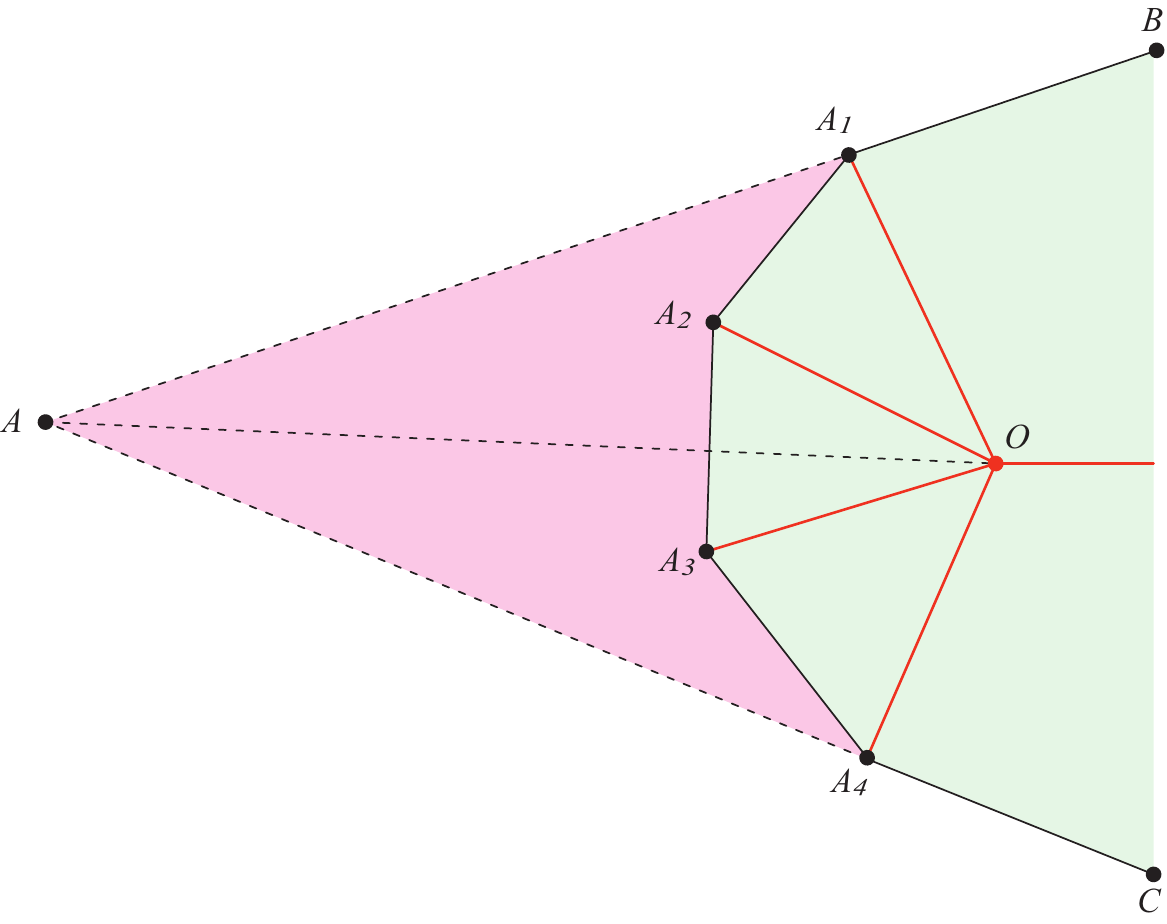}
\caption{Construction of congruent polygons.}
\label{f:fixedspeed}
\end{figure} 

Let $B$ and $C$ be vertices of the polygons adjacent to $A_1$ and $A_m$, respectively.  Let the intersection of lines $A_1 B$ and $A_m C$ be at point $A$; see Figure~\ref{f:fixedspeed}.   Consider the polygons in $P$ and $Q$ with vertices $AA_1 \ldots A_m$ (shaded red in Figure~\ref{f:fixedspeed}): the edges $A_iA_{i+1}$ are all equal, as are the angles subtended at vertices $A_i$ by the congruence of the triangles $OA_iA_{i+1}$.  So these polygons must be congruent in $P$ and $Q$, implying the angle $A_1AA_m$ and the length $OA$ be identical in both $P$ and $Q$ as well.  Deleting all $A_i$ vertices and adding vertex $A$ creates new convex polygons $P'$ and $Q'$ with new sides $BA$ and $AC$, both with identical angles at all leaves, with the same underlying skeleton.  By the induction hypothesis, they are congruent polygons, implying that $P$ and $Q$ cannot be distinct.
\end{proof}

\begin{thm} 
\label{t:rigidity_convex}
For $G \in \G$ with $n$ leaves, there are at most $2n-5$ suitable convex polygons.
\end{thm} 

\begin{proof}
It is sufficient to prove that for each possible choice of a chronological center for $G$, there exists at most one suitable convex polygon for $G$.  Since $G$ has $n-3$ interior edges and $n-2$ interior vertices, there are $2n-5$ possible chronological center choices.
Fix the chronological center in an edge or vertex of $G$, and let the vertices $v_i$ of $P$ move with velocities $\vel_i$. By Lemma~\ref{l:maxone}, there is at most one suitable convex polygon $P$.  

The angle $\alpha_i$ of the convex polygon subtended at $v_i$ is $2 \arcsin(1/\vel_i)$, resulting in
$$\sum_{i=1}^{n} \ \arcsin \left(\frac{1}{\vel_i}\right) \ = \ \frac{(n-2) \pi}{2} \, .$$
For contradiction, assume there is another set of velocities $\vel_i'$ resulting in a valid skeletal configuration with the same chronological center. Without loss of generality, assume $\vel_1' > \vel_1$.  Then by the Racing Lemma, $\vel_i' \geq \vel_i$ for all $i$, implying
$$\frac{(n-2) \pi}{2} \ = \ \sum_{i=1}^{n} \ \arcsin \left(\frac{1}{\vel_i'}\right) \ < \ \sum_{i=1}^{n} \ \arcsin \left(\frac{1}{\vel_i}\right) \ = \ \frac{(n-2) \pi}{2} \, ,$$
a contradiction.  
\end{proof}

%
%
\section{Computational issues} \label{s:comp}
\subsection{}

While the previous sections were focused on the rigidity of the skeletal configurations, nothing was said about how one would go about constructing a valid skeletal configuration when provided with a phylogenetic tree.  We begin with an algebraic proof of the rigidity of trivalent stars of Proposition~\ref{p:star-rigid}, helping pave the way for some discussion on the computational issues involved in the problem.

\begin{prop} 
\label{p:triangleunique}
There is a unique suitable triangle for every star $S_3$ of degree three.
\end{prop}

\begin{proof}
Assume an embedding of a star $S_3$ is given which is the straight skeleton of the triangle determined by the leaves.  Let the angles of the triangles be $2\alpha, 2\beta$, and $2\gamma$ at vertices $A, B$, and $C$ respectively, as in Figure~\ref{f:tri-alg}. 

\begin{figure}[h]
\includegraphics{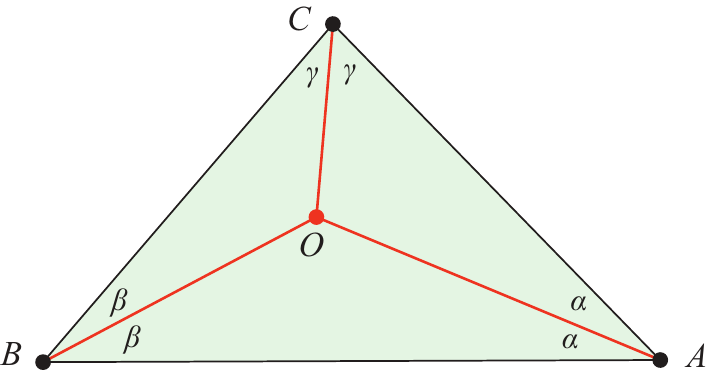}
\caption{Suitable triangle of a trivalent star.}
\label{f:tri-alg}
\end{figure}

The law of sines applied to triangle $AOC$ yields 
$$\frac{AO}{\sin \gamma} = \frac{CO}{\sin \alpha}.$$
Since $\gamma = \frac{\pi}{2} - \alpha - \beta$, it follows that 
$$\frac{AO}{\sin \gamma} = \frac{AO}{\cos(\alpha+\beta)} = \frac{AO}{\cos\alpha \cos\beta - \sin\alpha \sin\beta} = \frac{AO}{\sqrt{1-{\sin\beta}^2}\sqrt{1-{\sin\alpha}^2} - \sin\alpha \sin\beta}.$$ 
Thus 
$$\frac{AO}{\sqrt{1-{\sin\beta}^2}\sqrt{1-{\sin\alpha}^2} - \sin\alpha \sin\beta} =  \frac{CO}{\sin \alpha}.$$
The law of sines applied to triangle $AOB$ results in
$$\frac{AO}{\sin \beta} \ = \ \frac{BO}{\sin \alpha}\, , $$
which together with the previous equation produces the following polynomial:
\begin{equation} 
\label{eq:triangle} 
(2 AO^2 \cdot BO \cdot CO) \ x^3 + (AO^2 \cdot CO^2 + AO^2 \cdot BO^2 + BO^2 \cdot CO^2) \ x^2 - (CO^2 \cdot BO^2) = 0 
\end{equation}
Since the polynomial cannot vanish identically, we have three solutions, counting multiplicities. The discriminant of the equation is 
$$\Delta \ = \ -27 CO^4 \cdot BO^4 + (BO^2 + CO^2 (1 + BO^2))^3.$$ 
By the arithmetic mean -- geometric mean inequality, 
$$BO^2 + CO^2 (1 + BO^2) \ \geq \ 3 {(BO^2 \cdot CO^2 \cdot (CO^2 BO^2))}^\frac{1}{3}\, ,$$
showing $\Delta \geq 0$. So the solutions of the equations are all real, say $x_1, x_2, x_3$. By Vieta's formulas, we have the following system,
\begin{eqnarray*}
x_1 + x_2 + x_3 & = & -\frac{b}{a} \\
x_1 x_2 + x_1 x_3 + x_2 x_3 & = & 0 \\ 
x_1 x_2 x_3 & = & \frac{c}{a}
\end{eqnarray*}
showing that exactly two of the roots are negative, and exactly one (say $x_1$) is positive.  If $x_1 >1$, then the left hand side of Eq.~\eqref{eq:triangle} is positive, which cannot be as $x_1$ is a root.  Thus, $x_1 \leq 1$ must hold, being a valid value for $\sin \alpha$, and the only such value. 
\end{proof}

\subsection{}

The above approach can be pushed to give rise to a connection between our problem and the angle bisector problem, a geometric problem dating back to Euler, thoroughly studied in \cite{baker}:
\begin{ang}
Construct a triangle given the lengths of the angle bisectors.
\end{ang}
In the work of Baker \cite{zajic}, it is proven that this is impossible with ruler and compass, and the nature of the polynomial equations given the triangle lengths is studied. Our problem for the case of the triangle is slightly different, since the incenter is given as well.  But considering the straight skeletons of arbitrary polygons, it can be viewed as a generalization of the angle bisector problem. 

\begin{lem} 
The suitable triangle for star $S_3$ with integer edge lengths may not be constructible with ruler and compass.
\end{lem}

\begin{proof}
Let $O$ be the center of the star, and $A$, $B$, $C$ be the leaves, again as in Figure~\ref{f:tri-alg}. Consider the star edge lengths where $AO = r$, $BO = 2r$, $CO = 3r$, for $r \in \mathbb{N}$. Then Eq.~\eqref{eq:triangle} becomes 
$$12 r^3 + 49 r^2 - 36 = 0.$$ 
Any rational root $a/b$ of this equation is such that $a$ divides 36 and $b$ divides $12$; by exhaustion, no such combination succeeds.  But if a cubic equation with rational coefficients has no rational root, then none of its roots are constructible. So for these particular values, $\sin \alpha$ is not constructible.  
Thus, if the triangle was constructible, then clearly $\alpha$, as well as $\sin \alpha$, are constructible as numbers, a contradiction.
\end{proof}

\subsection{}

Consider the problem of constructing the suitable polygon for a feasible phylogenetic tree.  We close with discussing unsolvability issues of this problem when constrained to the algebraic model: only the arithmetic operations $+, -, \times, /$ along with $\sqrt[k]{}$ of rational numbers are allowed.

\begin{lem} 
Let $G \in \G$ be a phylogenetic tree with rational edge lengths.  In general, the convex suitable polygon $P$ for $G$ has side lengths not expressible by radicals over $\mathbb{Q}$. 
\end{lem} 

\begin{proof}
Consider a star $S_3$ with edge lengths $r, r, r, 10r/11, 10r/12$, for $r \in \mathbb{N}$.   For a suitable convex polygon $P$, the sum of the interior angles yields
\begin{equation}
\label{e:5star}
3\cdot \arcsin(x) + \arcsin\left(\frac{11x}{10}\right) + \arcsin\left(\frac{12x}{10}\right) \ = \ \frac{3 \pi}{2}\, ,
\end{equation}
where $x^{-1}$ equals the velocity $\vel$ at the leaf with edge length $r$. 
It is not hard to see that if Eq.~\eqref{e:5star} has a solution, there exists a suitable polygon $P$ for $G$.

Now we show that the edge lengths of $P$ are not expressible as radicals over $\mathbb{Q}$.  By rewriting $\arcsin x$ in its logarithmic form as $-i \ln(i x + \sqrt{1-x^2})$, and substituting into Eq.~\eqref{e:5star}, $x$ becomes the square root of one of the zeros of the polynomial
\begin{eqnarray*}
p(x) & = & 1 - 2330 x + 1837225 x^2 - 653926400 x^3 + 111607040000 x^4 \\
&& - 8795136000000 x^5 + 256000000000000 x^6.
\end{eqnarray*}
The discriminant of polynomial $p$ is
$$\Delta(p) = 2^{74} \cdot 3^6 \cdot 5^{30} \cdot 11^6 \cdot 23^2 \cdot 29 \cdot 31 \cdot 79^2 \cdot 43151 \cdot 2626069.$$
Since $p$ is irreducible modulo 13, and since 13 does not divide 1 (the constant term), $p$ is irreducible.  Considering $p$ modulo  $7$, $13$, and $17$ (the ``good'' primes), we get a (2+3)-permutation, a 5-cycle and a 6-cycle as the Galois groups, respectively, showing the Galois group of $p$ is the symmetric group ${\mathbb S}_6$ on six letters \cite[Lemma 8]{bajaj}.  Since ${\mathbb S}_6$ is not solvable, then $x = \frac{1}{\vel} = \sin \alpha$ is not expressible using  arithmetic operations and $\sqrt[k]{}$ of rational numbers, where $\alpha$ is the angle subtended at the vertex with velocity $\vel$. But since $\sin \alpha$ can be expressed as an expression involving arithmetic operations on radicals of the side lengths of polygon $P$, and the edge lengths of  tree $G$, the claim holds. 
\end{proof}

\begin{rem} 
Of course, the fact that the side lengths are not expressible via radicals over $\mathbb{Q}$ implies that the coordinates of the vertices of the polygon must also not be expressible via radicals over $\mathbb{Q}$. 
\end{rem}

\begin{open}
Construct an algorithm approximately calculating a polygon from a given phylogenetic tree.
\end{open}

\begin{open}
For a given phylogenetic tree, provide a rigidity result for arbitrary suitable polygons, not just convex ones.
\end{open}

%
%

\bibliographystyle{amsplain}

\end{document}